\documentclass[conference,10pt]{IEEEtran}
\IEEEoverridecommandlockouts
\usepackage{cite}
\usepackage{amsmath,amssymb,amsfonts}
\usepackage{algorithmic}
\usepackage{graphicx}
\usepackage{textcomp}
\usepackage{xcolor}
\usepackage{multicol}
\usepackage{cite}
\usepackage{amsmath}
\usepackage{amsfonts}
\usepackage{pifont}
\usepackage{amssymb}
\usepackage{amsthm}
\usepackage{tikz}
\usepackage{cases}
\usepackage{graphicx}
\usepackage{float,stfloats}
\graphicspath{{../}}
\DeclareGraphicsExtensions{.pdf}
\usepackage{epstopdf}
\usepackage{epsfig}
\usepackage{multirow}
\usepackage{booktabs}
\usepackage{xtab}
\usepackage{tabu}
\usepackage{longtable}
\usepackage{algorithm}
\usepackage{algorithmic}
\usepackage{enumerate}
\usepackage{makecell}
\usepackage{lipsum}
\usepackage{multicol}
\usepackage{mathdots}
\usepackage{extarrows}
\usepackage{color,xcolor}
\usepackage{bm}
\usepackage{arydshln}
\usepackage{multirow}
\usetikzlibrary{arrows}
\usepackage{verbatim}
\usepackage{cuted}
\usepackage{url}
\usepackage{textcomp}
\usepackage{balance}
  \usepackage[left=0.7in, right=0.7in, top=0.73in, bottom=1.06in]{geometry}
\theoremstyle{plain}

\newtheorem{theorem}{Theorem}

\newtheorem{proposition}[theorem]{Proposition}
\newtheorem{corollary}[theorem]{Corollary}

\theoremstyle{definition}
\newtheorem{definition}{Definition}

\newtheorem{construction}{Construction}

\newtheorem{remark}{Remark}

\renewcommand{\tablename}{\scriptsize TABLE}
\renewcommand{\figurename}{\scriptsize Fig.}

\def\BibTeX{{\rm B\kern-.05em{\sc i\kern-.025em b}\kern-.08em
    T\kern-.1667em\lower.7ex\hbox{E}\kern-.125emX}}

\begin{document}
\title{ Asymptotically Optimal Repair of Reed-Solomon Codes with Small Sub-Packetization under Rack-Aware Model
\thanks{This work was supported  by Key R\&D Program of Shandong Province, China (No. 2024CXPT066).}
}
 \author{\IEEEauthorblockN{Ke Wang\textsuperscript{*}, Zhongyan Liu\textsuperscript{\dag}, Rengang Li\textsuperscript{*}, Yaqian Zhao\textsuperscript{*}, and Yaqiang Zhang\textsuperscript{*}}
 \IEEEauthorblockA{\textit{\textsuperscript{*}IEIT SYSTEMS Co., Ltd., Jinan, 250101, China
} \\
{\textsuperscript{\dag }\textit{KLMM, Academy of Mathematics and Systems Science, Chinese Academy of Sciences, Beijing 100190, China}}\\
Emails: \{wangke12, lirg, zhaoyaqian, zhangyaqiang\}@ieisystem.com, liuzhongyan@amss.ac.cn}  
}
\maketitle
\begin{abstract}
This paper presents a comprehensive study on the asymptotically optimal repair of Reed-Solomon (RS) codes with small sub-packetization, specifically tailored for rack-aware distributed storage systems. Through the utilization of multi-base expansion, we introduce a novel approach that leverages monomials to construct linear repair schemes for RS codes. Our repair schemes which adapt to  all admissible parameters achieve asymptotically optimal repair bandwidth  while significantly reducing the sub-packetization compared with existing schemes. Furthermore, our approach is capable of repairing RS codes with asymptotically optimal repair bandwidth under the homogeneous storage model, achieving smaller sub-packetization than existing methods.
\end{abstract}
	
\section{Introduction}\label{Sec1}
Most distributed storage systems (DSS) use erasure codes, particularly Reed-Solomon (RS) codes\cite{facebook,baidu,google,mic}, to achieve maximum error correction capability under the given redundancy. 
Specifically, a file consisting of $k$ blocks is encoded into $n$ blocks by an $[n,k]$ RS code and distributed across different nodes, then the system can tolerate any $r=n-k$ nodes failing, where each block is composed of $l$ data units, and $l$ is called the sub-packetization.  Smaller $l$ implies lower implementation complexity, which is useful in practice. 

In DSS, single-node failures are the most common, so this paper focuses on the repair model for single-node failure scenarios. A key metric in node repair is the repair bandwidth (RB), defined as the total number of symbols downloaded to repair a failed node. Guruswami and Wootters \cite{Gu} first characterized linear repair schemes for RS codes using trace functions, which involve ``vectorizing" symbols in \( \mathbb{F}_{q^l} \) into \( l \) sub-symbols in \( \mathbb{F}_q \). Subsequent work aimed to design repair schemes for RS codes to minimize RB \cite{Dau1,Ye1,Ye2,Li,Vardy}.

Most research regarding the repair of RS codes concentrated on the homogeneous storage model, in which nodes are distributed uniformly in different locations. However, 
Modern large-scale data centers typically employ a hierarchical topology known as the rack-aware storage model. In this model, the $n$ storage nodes are organized into $\bar{n}$ racks and each rack contains $u$ nodes where $n={\bar n}u$. Assume $k={\bar k}u+v, v\equiv k \mod u$. 
We refer to a rack containing failed nodes as the host rack. During the repair process of failed nodes, there are two types of communications involved:
\begin{itemize}
    \item Intra-rack communication: the remaining $u - 1$ nodes in the host rack each transmit all the information they store to the replacement node;
    \item Inter-rack communication: in addition to the host rack, $\bar d (\geq \bar k)$ helper racks each transmit part of the information to the primary rack.
\end{itemize}
Moreover, the cross-rack communication cost is much more expensive than the intra-rack communication cost. Therefore, the intra-rack communication cost is negligible naturally and only the communication volume between racks in repairing failed nodes needs to be considered.  

Rack-aware model was first proposed by Hu et al. \cite{Hu1}, aiming to minimize the inter-rack RB of MDS codes. 
Under this model, a lower bound (cut-set bound) of inter-rack RB for $[n,k]$ MDS codes over $F=\mathbb{F}_{q^l}$ was derived in \cite{Hu1,Hou1}, i.e., RB $\geq\frac{\bar{n}-1}{\bar{n}-\bar{k}}l$. 
Subsequently, Chen et al. \cite{Chen1} explicitly constructed MDS codes that are  applicable to all parameters and have optimal RB (achieving the cut-set bound) under the rack-aware model.
For the single-node repair of RS codes, by generalizing the vectorized linear repair scheme under the homogeneous model, several types of RS codes were constructed to achieve the optimal RB under the rack-aware model. However, these codes have very limited parameters. Specifically, in  \cite{Jin}, the authors proposed two types of $[n,k]$ RS codes over $F$, requiring $n=|F|$ or $n=|F|-1$, and $\bar n-\bar k=\frac{l}{l-s}$, resulting in a low code rate $\frac{k}{n}<\frac{1}{2}$, where $s<l$. Additionally, the RS codes in \cite{wang2023} achieve optimal RB but demand an impractically large sub-packetization size $l \approx \bar n^{\bar n}$. 
Overall, these codes either exhibit large sub-packetization or low code rates.  

In this paper, we provide a tradeoff between the sub-packetization level and RB by multi-base expansion. Compared to the RS codes with optimal RB under the rack-aware model in \cite{wang2023}, the RS codes we construct achieve asymptotically optimal RB, but with significantly smaller sub-packetization, as summarized in Table \ref{t4.1}. Specifically, in our construction, $\bar r = \bar n - \bar k \triangleq \prod_{i=1}^{m} p_i$, $\{p_i\}_{i \in [m]}$ are all prime numbers. When $\bar r$ is not a prime (i.e., $m \geq 2$), Construction \ref{cons2} yields the smallest sub-packetization level among the compared schemes. The main contributions of this paper are formalized in Theorem \ref{thm6} and Corollary \ref{cor7}.

\begin{table}[ht]
\centering
\renewcommand\arraystretch{1.25}
\caption{\scriptsize Comparison of the Sub-Packetizetion of  RS Codes with Optimal/Asymptotically Optimal RB.} \label{t4.1}
\resizebox{0.49\textwidth}{!}
{
\footnotesize\begin{tabular}{|l|c|c|}
\hline
  RS codes & Optimal/Asymptotically Optimal & Sub-packetization $l$ \\\hline
  Reference \cite{wang2023} & Optimal & $l \approx \bar{n}^{\bar{n}}$\\ \hline
  Construction \ref{cons1}& Asymptotically Optimal & $l = \bar{r}^{\bar{n}}$\\ \hline
  Construction \ref{cons2}& Asymptotically Optimal & $l \leq \bar{r}^{\lceil{\scalebox{0.6}{$\frac{\bar{n}}{m}$}}\rceil}$\\ \hline
 \end{tabular}
 }
\end{table}

\begin{remark} Specially, if we set $u=1$, i.e., $\bar{n}=n$ and $\bar{r}=r$, then our scheme can be applied to repair RS codes under the homogeneous storage mode. Compared with the RS codes in \cite{Vardy}, which assume $r=s^{m'}$ and $l=s^{m'+n-1}=r^{\frac{m'+n-1}{m'}}$, our codes offer more flexible parameters and smaller sub-packetization. Specifically,
\begin{itemize}
\item if $r=s^{m'},m'\geq 2$, and $s$ is a prime number, our sub-packetization is $l=r^{\frac{n}{m'}}=s^{n}<s^{m'+n-1}$.
\item if $r$ is a prime number, our scheme achieves lower RB than the scheme in \cite{Vardy}. For example, when $r=7$, we can utilize $2\times 3$ to approximate $7$, whereas in \cite{Vardy}, they can only utilize $2^2$ to approximate $7$. Therefore, the RB of our scheme is approximately $\frac{7}{6}$ times the optimal RB, while the RB in \cite{Vardy} is approximately $\frac{7}{4}$ times.
\item if $r=s^{m'},m'\geq 1$, and $s=\prod_{j=1}^{\tilde{m}}p_j$, where $p_j$'s are all prime numbers and $\tilde{m}\geq 2$, then our sub-packetization is $l\leq r^{\lceil{\frac{n}{m'\tilde{m}}}\rceil}<r^{\frac{m'+n-1}{m'}}=s^{m'+n-1}$, where the second inequality is due to $\frac{m'+n-1}{m'}\geq\lceil{\frac{n}{m'}}\rceil$.

\end{itemize}   
\end{remark}


\section{Preliminaries}\label{Sec2}
\subsection{Some Notations}\label{Sec2.1}
For nonnegative integers $i<j$, denote $[j]=\{1,...,j\}$ and $[i,j]=\{i,i+1,...,j\}$. Let $B=\mathbb{F}_q$ be the finite field of $q$ elements and $F=\mathbb{F}_{q^l}$ be a field extension of $B$. 
\begin{definition}[RS code]
The Reed-Solomon code $RS(A,k)$ of dimension $k$ over $F$ with evaluation points $A=\{\alpha_1, \alpha_2 ,\dots,\alpha_n \}\subseteq F$ is defined by
\begin{equation*}
	RS(A,k)=\{(f(\alpha_1),...,f(\alpha_n)):f\in F[x], \mathrm{deg}(f)\leq k-1\}.
\end{equation*} 
\end{definition}

\subsection{Guruswami-Wootters's Repair Scheme}\label{Sec2.3}
In \cite{Gu},
Guruswami and Wootters  gave a characterization of linear repair schemes for RS codes. For any $\alpha\in F$, the trace $\mathrm{tr}_{F/B}(\alpha)$ of $\alpha$ over $B$ is $\mathrm{tr}_{F / B}(\alpha)=\alpha+\alpha^q+\cdots +\alpha^{q^{l-1}}.$ Let $\{\zeta _1,\dots ,\zeta _l\}$ and $\{\mu _1,\dots,\mu _l \}$ be two bases of ${F}$ over ${B}$. Then they are said to be dual bases if $\mathrm{tr}_{F/B}(\zeta_i\mu_j)= \delta_{i,j}$. For any $\alpha\in F$, it is well known that $\alpha=\sum_{i=1}^l{\mathrm{tr}_{F/B}(\zeta _i\alpha)\mu _i}.$ Actually, $\{ \mathrm{tr}_{F/B}(\zeta _i\alpha)\} _{i=1}^{l}$ uniquely determines $\alpha$.

Let $\mathcal{C}=RS(A,k)$ be an $[n,k]$ RS code defined over $F$. Suppose the symbol $c_i$ in a codeword ${\bm c}=( c_1,\dots ,c_n)\in \mathcal{C}$ is erased. The repair scheme, as detailed in \cite{Gu}, can be summarized as follows:
\begin{enumerate}
\item Let $\mathcal{C}^{\bot}$ be the dual code of $\mathcal{C}$, i.e., $\mathcal{C}^{\bot}=\{(\lambda_1g(\alpha_1),...,\lambda_ng(\alpha_n)):g\in F[x], \mathrm{deg}(g)\leq n-k-1\},$
where $\lambda_i=\prod_{j\neq i}(\alpha_i-\alpha_j)^{-1}$.	
	  
\item Find $l$ dual codewords $\{{\bm c}_j^{\bot}=
(\lambda_1g_j(\alpha_1),...,\lambda_ng_j(\alpha_n)):j\in[l]\}\subseteq \mathcal{C} ^{\bot}
$ such that $\{g_1(\alpha_i),\dots,g_l(\alpha_i)\} $ forms a basis of $F$ over $B$. 
\item $c_i$ can be recovered using $\{{\rm tr}_{F/B}(\lambda_tg_j(\alpha_t)c_t)\}_{j\in[l]}$ provided by the surviving node $t\neq i$. The RB is 
\begin{equation*}
	b=\sum_{t\in[n]\setminus\{i\}}{\mathrm{rank}_B(\{ g_j(\alpha_t)\}_{j\in [l]})}.
\end{equation*}
Specifically, since ${g_1(\alpha_i), \dots, g_l(\alpha_i)}$ forms a basis of $F$ over $B$, then $\{ \mathrm{tr}_{F/B}(\lambda_ig_j(\alpha_i)c_i)\} _{j=1}^{l}$ is sufficient to recover $c_i$. Given that ${\bm c}_j^{\bot}\in\mathcal{C}^{\bot}$, it follows $\lambda_ig_j(\alpha_i)c_i=-\sum_{t\in[n]\setminus\{i\}}\lambda_tg_j(\alpha_t)c_t$, and thus ${\rm tr}_{F/B}\big(\lambda_ig_j(\alpha_i)c_i\big)$ can be obtained from
\end{enumerate}
$$ {\rm tr}_{F/B}\big(\lambda_ig_j(\alpha_i)c_i\big)=-\sum_{t\in[n]\setminus\{i\}}{\rm tr}_{F/B}\big(\lambda_tg_j(\alpha_t)c_t\big), \forall j\in[l].$$
In summary, a linear repair scheme for a failed node $i$ in $RS(A,k)$ can be characterized by $l$ polynomials $\{g_j(x)\}_{j\in[l]}$ such that ${\rm rank}_B\big(\{g_j(\alpha_i)\}_{j\in[l]}\big)=l$. 

\subsection{Rack-Aware Distributed Storage System}
Consider an $[n,k]$ RS code $\mathcal{C}$ defined over $F$. Suppose that $n=\bar{n}u,k=\bar{k}u+v$, where $u\mid n$ and $v=k \mod u$ with $v\in[0,u-1]$. Denote $r=n-k$ and $\bar{r}=\bar{n}-\bar{k}$. To avoid the trivial case, we assume $k \geq u$ throughout this paper.

In rack-aware storage systems, the $n = \bar{n}u$ nodes are organized into $\bar{n}$ racks, each containing $u$ nodes. Then, each node can be labeled with pair $(e, m)$, representing the 
$m$-th node in the $e$-th rack. Subsequently, for any codeword $(f(\alpha_1),f(\alpha_2)\dots,f(\alpha_n))$ in $\mathcal{C}$, we represent it as: $$(f(\alpha_{1,1}),\dots,f(\alpha_{1,u}),\dots,f(\alpha_{\bar{n},1}),\dots,f(\alpha_{\bar{n},u}))$$
where $\alpha_{e,m}$ corresponds to the evaluation point for $m$-th node in the $e$-th rack. 
Suppose $\{g_j(x)\}_{j=1}^{l}$ defines a linear repair scheme for some failed node in rack $i$. Since the inner rack bandwidth is negligible in rack-aware storage systems, the RB of the scheme defined by $\{g_j(x)\}_{j=1}^{l}$ is 
\begin{equation}\label{bformula}
b=\sum_{e\in[\bar n]\setminus\{i\}}b_e,
\end{equation}    
where $b_e=\mathrm{rank}_B(\{g_j(\alpha_{e,m}):m\in[u],j\in[l]\})$. A lower bound on the RB, also known as the cut-set bound, is given by \cite{Hu1,Hou1} as follows:
\begin{equation}\label{bbound}
	b\geq\frac{(\bar{n}-1)l}{\bar{r}}.
\end{equation}	
For convenience, we denote $b_{\min}=\frac{(\bar{n}-1)l}{\bar{r}}.$

\section{Repairing Reed-Solomon Codes via Monomials}\label{Sec3}
In this section, we present a linear repair scheme for RS codes utilizing monomials. The RS codes in this section are suitable for all admissible parameters. Furthermore, compared with the RS codes in \cite{Chen1} that achieve optimal repair bandwidth, our codes offer asymptotically optimal RB with significantly reduced sub-packetization.

\subsection{Repairing RS Codes for All Admissible Parameters}\label{Sec3.1}
Set $l=\bar r^{\bar n}$. Let $\zeta$ be a primitive element of $F$, then the order of $\zeta$ is $q^l-1$, denoted by $o(\zeta)=q^l-1$. 
\begin{proposition}\label{pro7}
For any	$u>0$ such that $u\mid q-1$, it has ${\rm rank}_B\big(\{1,\zeta^u,\zeta^{2u},\dots,\zeta^{(l-1)u}\}\big)=l$.
\end{proposition}
\begin{proof}
Suppose that $B(\zeta^u)=\mathbb{F}_{q^t}$, then $t\mid l$ and $(\zeta^u)^{q^t-1}=1$. Since $o(\zeta)=q^l-1$, it follows that $(q^l-1)\mid u(q^t-1)$. Combining with $u\mid q-1$, we know $\frac{q^l-1}{q^t-1}\leq u\leq q-1$, which implies that $t=l$. Therefore, $\{1,\zeta^u,\zeta^{2u},\dots,\zeta^{(l-1)u}\}$ are linearly independent over $B$. The proof is completed.
\end{proof}

Next, we introduce a family of $[n,k]$ RS code $RS(A,k)$  and the corresponding linear repair schemes which can be  derived from the construction in \cite{Ye1} easily. As beginning, for any $t\in[0,l-1]$, denote its $\bar r$-ary expansion as $t=(t_1,t_2,\dots,t_{\bar n})$, i.e., $t=\sum_{i=1}^{\bar{n}}t_i\bar{r}^{i-1}$, where  $t_i\in[0,\bar{r}-1]$.
\begin{construction}\label{cons1}
We first define the evaluation points $A$. Suppose $u\mid q-1$. Choose some $\alpha\in B$ such that $o(\alpha)=u$. For $i\in[\bar{n}]$, let $A_i=\{\alpha_{i,1},\alpha_{i,2},\dots,\alpha_{i,u}\}$, where $\alpha_{i,j}=\zeta^{\bar{r}^{i-1}}\alpha^j$ for $j\in[u]$. It is obviously that $A_1,...,A_{\bar{n}}$ are pairwise disjoint. Define $A=\bigcup_{i=1}^{\bar n}A_i$, it has $|A|=\bar nu=n$.

Then, we define the repair polynomials $\{g_{t,s}(x):t\in T_i,s\in[0,\bar r-1]\}$ for failed node $\alpha_{i,j}$ in rack $i$, $i\in[\bar n],j\in[u]$. 
Denote $T_i=\{t\in[0,l-1]:t_i=0\}$. We have  $|T_i|={\bar r}^{\bar{n}-1}$. For $t\in T_i$ and $s\in[0,\bar r-1]$, define 
\begin{equation}\label{eq8}
	g_{t,s}(x)=\zeta^{ut}x^{us}. 
\end{equation}
It is clear that $\deg({g_{t,s}(x)})\leq u\bar r-u\leq r-1$.
\end{construction}    
We first check that $\{g_{t,s}(x):t\in T_i,s\in[0,\bar r-1]\}$ satisfies the repair condition.
It can be seen $g_{t,s}(\alpha_{i,j})=\zeta^{ut}\zeta^{us\bar r^{i-1}}\alpha^{jus}=(\zeta^{u})^{t+s\bar r^{i-1}}$, where the last equality is due to $o(\alpha)=u$. Consequently,
\begin{align*}
	&\quad \ \{g_{t,s}(\alpha_{i,j}):t\in T_i,s\in[0,\bar r-1]\}\\
	&=\{(\zeta^{u})^{t+s\bar r^{i-1}}:t\in T_i,s\in[0,\bar r-1]\}\\
	&=\{(\zeta^{u})^a:a\in[0,l-1]\}.
\end{align*}
It immediately follows from Proposition \ref{pro7} that $\mathrm{rank}_{B}(\{g_{t,s}(\alpha_{i,j}):t\in T_i,s\in[0,\bar r-1]\})=l$. Therefore, $\{g_{t,s}(x):t\in T,s\in[0,\bar r-1]\}$ defines a repair scheme for node $\alpha_{i,j}$. We then proceed to calculate the RB of this scheme.                                   
\begin{theorem}\label{thm8}
The repair bandwidth, denoted by $b$, of the linear repair scheme for RS$(A,k)$ defined in Construction \ref{cons1} satisfies $b<\frac{(\bar n+1)l}{\bar r}$.	
\end{theorem}
\begin{proof}
Suppose node $\alpha_{i,j}$ in rack $i$ fails, $i\in[\bar n],j\in[u]$. Then, $\{g_{t,s}(x):t\in T_i,s\in[0,\bar{r}-1]\}$, defined as in Construction \ref{cons1}, characterizes a linear repair scheme for node $\alpha_{i,j}$.
According to (\ref{bformula}), the RB is $b=\sum_{e\in[\bar n]\setminus\{i\}}b_e$, where $$b_e=\mathrm{rank}_B(\{g_{t,s}(\alpha_{e,j}):j\in[u],t\in T_i,s\in[0,\bar r-1]\}).$$
For any $j\in[u]$, $g_{t,s}(\alpha_{e,j})=(\zeta^{u})^{t+s\bar r^{e-1}}$, thus $b_e=\mathrm{rank}_B\{(\zeta^{u})^{t+s\bar r^{e-1}}:t\in T_i,s\in[0,\bar r-1]\}.$
Then we can do some similar discussion  in \cite{Ye1} as follows:
\begin{itemize}
\item[1)] When $e<i$,  
\begin{align*}&\{(\zeta^{u})^{t+s\bar r^{e-1}}:t\in T_i,s\in[0,\bar r-1]\}=\{\zeta^{ut}:t\in T_i\}\cup\\
	&\big(\bigcup_{s=0}^{\bar r-2}\{\zeta^{ut}:t_i=1, t_e=s, t_{e+1}=\cdots=t_{i-1}=0\}\big).
		\end{align*} Hence
	$b_e\leq \frac{l}{\bar r}+(\bar r-1)\frac{l}{\bar r^{i-e+1}}$.
\item[2)] When $e>i$,
\begin{align*}&\{(\zeta^{u})^{t+s\bar r^{e-1}}:t\in T_i,s\in[0,\bar r-1]\}=\{\zeta^{ut}:t\in T_i\}\cup\\
	&\big(\bigcup_{s=0}^{\bar r-2}\{\zeta^{u(l+t)}:t_i=0, t_e=s, t_{e+1}=\cdots=t_{\bar n}=0\}\big).
\end{align*}
Hence $b_e\leq \frac{l}{\bar r}+(\bar r-1)\frac{l}{\bar r^{\bar n-e+2}}$.
\end{itemize}
Therefore, an upper bound on the sum of the dimensions is given by:
$b=\sum_{e\in[\bar n]\setminus\{i\}}b_e<\frac{(\bar n+1)l}{\bar r}.$\end{proof}
  It is observed that $\frac{b}{b_{\min}}\rightarrow 1$ as $\bar n\rightarrow\infty$, indicating that our repair scheme approaches optimality.

\subsection{Further Reducing The Sub-Packetization Level}\label{Sec3.2}
In this section, we utilize the multi-base expansion which will be defined in Theorem \ref{thm9} to further reduce the sub-packetization $l$ of RS$(A,k)$ in Construction \ref{cons1} while maintaining the asymptotically optimal repair property.
Suppose that $\bar r=\bar n-\bar k=p_1p_2\cdots p_m$ and $m\geq 2$\footnote{The requirement that $m\geq 2$ here is to distinguish the new construction from Construction \ref{cons1}, and this restriction will be dropped later.}, where $\{p_i\}_{i\in[m]}$ are all primes.  
\begin{theorem}[Multi-base expansion]\label{thm9}
For any $a\in [0, \bar r-1]$, there exists unique $\{a_i\}_{i\in[m]}$ such that \begin{equation}\label{eq9}
a=a_1+a_2p_1+\cdots+a_i\prod_{j=1}^{i-1}{p_j}+\cdots+a_m\prod_{j=1}^{m-1}{p_j}, 
\end{equation}where $a_i\in[0,p_i-1]$. 
\end{theorem}
\begin{proof}
First, we suppose \eqref{eq9} holds and prove the uniqueness of the expression of $a$. For some $a\in[0, \bar r-1]$, suppose that there exist other integers  $\{a_i'\}_{i\in[m]}$ such that 
\begin{equation}\label{a}
a=a_1'+a_2'p_1+\cdots+a_i'\prod_{j=1}^{i-1}{p_j}+\cdots+a_m'\prod_{j=1}^{m-1}{p_j}.    
\end{equation} Let $y$ be the smallest integer in $[m]$ satisfying $a_y\ne a_y'$. We may assume $a_y>a_y'$. Combining (\ref{eq9}) and (\ref{a}), it has 
\begin{equation}\label{eq6}
 (a_y-a_y')\prod_{j=1}^{y-1}{p_j}=\sum_{i=1}^{y-1}(a_i'-a_i)\prod_{j=1}^{i-1}{p_j}.   
\end{equation}
Since $a_i,a_i'\in[0,p_i-1]$, we know $\sum_{i=1}^{y-1}(a_i'-a_i)\prod_{j=1}^{i-1}{p_j}\leq \sum_{i=1}^{y-1}(p_i-1)\prod_{j=1}^{i-1}{p_j}=\big(\prod_{j=1}^{y-1}{p_j}\big)-1.$ However,
$\big(\prod_{j=1}^{y-1}{p_j}\big)-1<\prod_{j=1}^{y-1}{p_j}\leq(a_y-a_y')\prod_{j=1}^{y-1}{p_j}$, which contradicts (\ref{eq6}).
		
Then, we establish the validity of the expression in \eqref{eq9}. Denote the right hand of \eqref{eq9} by $\bm a=(a_1,a_2,\dots,a_m)$. From the uniqueness proved previously, $\{\bm a:a_i\in[0,p_i-1], i\in[m]\}$ can represent $\prod_{i=1}^{m}p_i=\bar{r}$ distinct integers. Since $$0\leq\sum_{i=1}^{m}a_i\prod_{j=1}^{i-1}{p_j}\leq \sum_{i=1}^{m}(p_i-1)\prod_{j=1}^{i-1}{p_j}=\big(\prod_{j=1}^{m}{p_j}\big)-1=\bar r-1,$$ the proof is completed.
\end{proof}
\begin{remark}
Actually, Theorem \ref{thm9} does not depend on the order of these primes.    
\end{remark}

Set $l=\bar r^{\lfloor{\frac{\bar n}{m}}\rfloor}\prod_{j=1}^{h}p_j\leq r^{\lceil{\frac{\bar n}{m}}\rceil}$, where $h\equiv\bar n\mod m$. Let $\zeta$ be a primitive element of $F$, then $o(\zeta)=q^l-1$. 

For simplicity, we consider $m\mid \bar n$ first, i.e., $l=\bar r^{n'}$, where $n'=\frac{\bar n}{m}$. To rule out the trivial case, we assume $n'\geq 2.$ Denote $d_{w,y}=\bar r^{w}\prod_{j=0}^{y-1}{p_j}$, where $p_0=1$. For any $t\in[0,l-1]$, it follows from Theorem \ref{thm9} there exists unique $\{t_1,...,t_{n'm}\}$ such that $$t=\sum_{w=0}^{n'-1}\sum_{y=1}^mt_{wm+y}d_{w,y},$$
where $t_{wm + y }\in[0,p_{y}-1]$. We refer to  $$(t_{1},\dots,t_{m},\dots,t_{(n'-1)m+1},\dots,t_{n'm})$$ as the multi-base expansion of $t$. 

Next, we present the construction of $[n,k]$ RS code $RS(A,k)$ and the corresponding linear repair schemes.
\begin{construction}\label{cons2}
We first define the evaluation points $A$. Suppose $u\mid q-1$. Choose some $\alpha\in B$ such that $o(\alpha)=u$. For $w\in[0,n'-1], y\in[m]$, let $A_{w,y}=\{\alpha_{w,y,1},\alpha_{w,y,2},\dots,\alpha_{w,y,u}\}$, where $\alpha_{w,y,j}=\zeta^{{d_{w,y}}}\alpha^j$ for $j\in[u]$. It is clear that $A_{0,1},...,A_{n'-1,m}$ are pairwise disjoint. Define $A=\bigcup_{w\in[0,n'-1], y\in[m]}A_{w,y}$, it has $|A|=n'mu=\bar nu=n$. 	

Then, we define the repair polynomials $\{g_{t,s}(x):t\in T_i,s\in[0,\bar r-1]\}$ for failed node $\alpha_{w,y,j}$ in rack $(w,y)$, $w\in[0,n'-1], y\in[m],j\in[u]$. If $w\in[0,n'-2]$ define
$$T_{w,y}=\{t\in[0,l-1]:t_{wm+y+y'}=0,y'\in[0,m-1]\},$$
if $w=n'-1$, define $$ T_{n'-1,y}=\{t\in[0,l-1]:t_{\overline{(n'-1)m+{y+y'}}}=0,y'\in[0,m-1]\},$$
where
$$ \overline{(n'-1)m+{y+y'}}=\begin{cases}
	(n'-1)m+{y+y'}  &\!\! \text{if} ~ y+y'\leq m \\
    {y+y'} \mod m  &\!\! \text{otherwise}
\end{cases}.$$
For $t\in T_{w,y}$ and $s\in[0,\bar r-1]$, define 
\begin{equation}
	g_{t,s}(x)=\zeta^{ut}x^{us}. 
\end{equation}
It is clear that $\deg({g_{t,s}(x)})\leq u\bar r-u\leq r-1$.
           
\end{construction}

\begin{remark}
This construction is also suitable for the case $m\nmid \bar n$. When $m\nmid \bar n$, set $l=\bar r^{\lfloor{\frac{\bar n}{m}}\rfloor}\prod_{j=1}^{h}p_j$, where $h=\bar n\mod m$. For any $t\in[0,l-1]$, we can transform its multi-base expansion to $$(t_{1},\dots,t_{(\lfloor{\frac{\bar n}{m}}\rfloor-1)m+1},\dots,t_{\lfloor{\frac{\bar n}{m}}\rfloor m}, t_{\lfloor{\frac{\bar n}{m}}\rfloor m+1},\dots,t_{\lfloor{\frac{\bar n}{m}}\rfloor m+h}),$$
where $t_{wm+y}$ is the $(wm + y )$-th digit corresponding to $d_{w,y}$ as defined previously. 

The evaluation points $A$ can be selected as follows. Suppose $u\mid q-1$. Choose some $\alpha\in B$ such that $o(\alpha)=u$. For $w\in[0,\lfloor{\frac{\bar n}{m}}\rfloor], y\in[m]$, let $A_{w,y}=\{\alpha\zeta^{{d_{w,y}}}, \alpha^2\zeta^{{d_{w,y}}},\dots, \alpha^{u}\zeta^{{d_{w,y}}}\}.$
Define $$A=(\bigcup_{w\in[0,\lfloor{\frac{\bar n}{m}}\rfloor-1],y\in[m]}A_{w,y})\cup(\bigcup_{y\in[h]}A_{\lfloor{\frac{\bar n}{m}}\rfloor,y}),$$ it has $|A|=(m\lfloor{\frac{\bar n}{m}}\rfloor+h)u=\bar nu=n$.

The repair polynomials are analogous to those in Construction \ref{cons2}.
\end{remark}
The following Proposition \ref{pro11} and Proposition \ref{pro12} indicate the polynomials constructed in Construction \ref{cons2} can be used to repair the failed nodes.
The proofs of Proposition \ref{pro11} and Proposition \ref{pro12} are provided in the Appendix.

\begin{proposition}\label{pro11}
If $w\in[0,n'-2]$, $\mathrm{rank}_{B}(\{g_{t,s}(\alpha_{w,y,j}):t\in T_{w,y},s\in [0,\bar r-1]\})=l$ for any $y\in[m]$ and $j\in[u]$, where $T_{w,y}$ is defined as in Construction \ref{cons2}.
\end{proposition}

\begin{proposition}\label{pro12}
If $w=n'-1$, $\mathrm{rank}_{B}(\{g_{t,s}(\alpha_{w,y,j}):t\in T_{w,y},s\in [0,\bar r-1]\})=l$ for any $y\in[m]$ and $j\in[u]$, where $T_{n'-1,y}$ is defined as in Construction \ref{cons2}.
\end{proposition}
Subsequently, we determine the repair bandwidth. 
\begin{theorem}\label{thm6}
For $w\in[0,n'-1], y\in[m],j\in[u]$, the repair bandwidth of the linear repair scheme for node $\alpha_{w,y,j}$ of RS$(A,k)$ defined in Construction \ref{cons1} satisfies:
\begin{itemize}
	\item[{\rm(i)}] If $w\in[n'-3]$, \begin{equation}\notag
		b<\frac{(\bar n-1)+3\bar r+3m+1}{\bar r}l\;.
	\end{equation}
	\item[{\rm(ii)}] If $w=n'-2$, \begin{equation}\notag
		b<\frac{(\bar n-1)+(m+3)\bar r+4}{\bar r}l\;.
	\end{equation}
	\item[{\rm(iii)}] If $w=n'-1$, $$b<\frac{(\bar n-1)+(m+1)\bar r+2}{\bar r}l\;.$$ 
\end{itemize}
Then for a fixed $\bar r$, as $\bar n \rightarrow \infty$, $\frac{b}{b_{\min}} \rightarrow 1$, where $b_{\min} = \frac{(\bar n-1)l}{\bar r}$.
\end{theorem}
\begin{proof}
Due to the limited space, we only prove the case of $w=n'-1$ and the proofs of case (i) and case (ii) can be analogized to the proof of case (iii). Let $T_{n'-1,y}$ and $g_{t,s}(x),t\in T_{n'-1,y},s\in[0,\bar r-1]$ be defined as in Construction \ref{cons2}. It follows that $|T_{n'-1,y}|=\frac{l}{\bar r}$. Then the repair bandwidth is $b=\sum_{(w',z)\neq(n'-1,y)}b_{w',z},$ where
\begin{align*}
	b_{w',z}&=\mathrm{rank}_B\{g_{t,s}(\alpha_{w',z,j}):t\in T_{n'-1,y},s\in[0,\bar r-1]\}\\
	        &=\mathrm{rank}_B\{(\zeta^{u})^{t+sd_{w',z}}:t\in T_{n'-1,y},s\in[0,\bar r-1]\}.
\end{align*}
Denote $V_{w',z}=\{(\zeta^{u})^{t+sd_{w',z}}:t\in T_{n'-1,y},s\in[0,\bar r-1]\}$. We can establish upper bounds on $b_{w',z}$ in three cases according to $w'm+z$.
\begin{enumerate}
\item[1)] $w'm+z\in[y,(n'-1)m+y-m]$. 
Define \begin{align*}S_{w',z}&=\{t\in[0,l-1]:t_1=\cdots=t_{y-1}=0,\\&t_{w'm+z+m}=\cdots=t_{(n'-1)m+y-1}=0,\\&t_{(n'-1)m+y}=1,t_{(n'-1)m+y+1}=\cdots=t_{n'm}=0\}.\end{align*} Then, $|S_{w',z}|\!\!\leq \!\!\!\frac{l}{\bar{r}^{\rho}\prod_{j=0}^{\mu}p_{j}},$ where $\rho\!\!=\!\!\!\lfloor{\frac{(n'-1)m+y-w'm-z}{m}}\rfloor$ and $\mu=((n'-1)m+y-w'm-z)\mod m$.
It can be checked that $V_{w',z}\subseteq \{(\zeta^u)^{t'}:t'\in T_{n'-1,y}\cup S_{w',z}\}$ and $|S_{w',z}|\leq \frac{l}{\bar{r}^{\rho}\prod_{j=0}^{\mu}p_{j}}$, where $\rho=\lfloor{\frac{(n'-1)m+y-w'm-z}{m}}\rfloor$ and $\mu=((n'-1)m+y-w'm-z)\mod m$. Thus, 
\begin{equation}\label{b1}
	b_{w',z}\leq |T_{n'-1,y}|+|S_{w',z}|\leq \frac{l}{\bar r}+\frac{l}{\bar{r}^{\rho}\prod_{j=0}^{\mu}p_{j}}.
\end{equation}

\item[2)]$w'm+z\in[1,y-1]$, i.e., $w=0$,  if $y-1<1$, we stipulate that $[1,y-1]$ is an empty set. Define \begin{align*}T_1=\{&t\in[0,l-1]:t_1=\cdots=t_{z-1}=0,\\&t_{(n'-1)m+y}=\cdots=t_{n'm}=0\},\end{align*}
	\begin{align*}	T_2 =\{&t\in[0,l-1]:t_1=\cdots=t_{z-1}=0,\\&t_{z+m}=\cdots=t_{(n'-1)m+y-1}=0,\\&t_{(n'-1)m+y}=1,t_{(n'-1)m+y+1}=\cdots=t_{n'm}=0\}.\end{align*}
Then, $|T_1\cup T_2|\leq \frac{l}{\prod_{j=0}^{m+z-y}p_{j}}+\frac{l}{\bar r^{n'-1}}$. It can be checked that $V_{w',z}\subseteq \{(\zeta^u)^{t'}:t'\in T_{n'-1,y}\cup T_1\cup T_2\}$. Thus, 
\end{enumerate}
\begin{equation}\label{b2}
b_{w',z}\leq |T_{n'-1,y}|+|T_1\cup T_2|\leq \frac{l}{\bar r}+\frac{l}{\prod_{j=0}^{m+z-y}p_{j}}+\frac{l}{\bar r^{n'-1}}.
\end{equation}

\begin{enumerate}
\item[3)]$w'm+z\in [(n'-1)m+y-m+1,n'm]\setminus\{(n'-1)m+y\}$.
	\begin{enumerate}
		\item[3.1)] $w'm+z\in [(n'-1)m+y-m+1,(n'-1)m]$. 
		
Define 
\begin{align*}
	\tilde{T_1}=\{&t\in[0,l-1]:t_{1}=\cdots=t_{y-1}=0,\\&t_{w'm+z+m}=\cdots=t_{n'm}=0\},
\end{align*}
\begin{align*}
\tilde{T_2}=\{&t\in[0,l-1]:t_{1}=\cdots=t_{y-1}=0,\\ &t_{(n'-1)m+y}=\cdots=t_{w'm+z+m-1}=0,\\&t_{w'm+z+m}=1,t_{w'm+z+m+1}=\cdots=t_{n'm}=0\}.
\end{align*}
Then, $V_{w',z}\subseteq\{(\zeta^u)^{t'}:t'\in T_{w',y}\cup \tilde{T_1}\cup\tilde{ T_2}\} $ and 
\end{enumerate}
\end{enumerate}
\begin{equation}\label{b31}
b_{w',z}\leq |T_{w',y}|+|\tilde{T_1}\cup\tilde{ T_2}|\leq\frac{2l}{\bar r}+\frac{l}{\prod_{j=0}^{(n'-w'-1)m+y-z}p_{j}}.	
\end{equation}

\begin{enumerate}
\item[3.2)] $w'm+z\in [(n'-1)m+1,(n'-1)m+y-1]$, i.e., $w'=n'-1$. 
Define  
\begin{align*}
	&\tilde{T_1}'=\{t\in[0,\big(\prod_{j=0}^{z-1}p_jl\big)-1]:t_{1}=\cdots=t_{y-1}=0\},\\
	&\tilde{T_2}'=\{t\in[0,\big(\prod_{j=0}^{z}p_jl\big)-1]:t_{1}=\cdots=t_{y-1}=0,\\&t_{(n'-1)m+y}=\cdots=t_{w'm+z+m-1}=0,t_{w'm+z+m}=1\}.
\end{align*}
Then, $V_{w',z}\subseteq\{(\zeta^u)^{t'}:t'\in T_{w',y}\cup \tilde{T_1}'\cup\tilde{T_2}'\} $ and 
\begin{equation}\label{b32}
b_{w',z}\leq |T_{w',y}|+|\tilde{T_1}'\cup\tilde{T_2}'|\leq\frac{2l}{\bar r}+\frac{l}{\prod_{j=0}^{y-z}p_{j}}.
\end{equation}

\item[3.3)]$w'm+z\in[(n'-1)m+y+1,n'm]$. 
One can derive that \begin{equation}\label{b33}
	b_{w',z}\leq l.\end{equation}
\end{enumerate}

Combining with \eqref{b1}, \eqref{b2}, \eqref{b31}, \eqref{b32} and \eqref{b33}, the total RB for $w=n'-1$ is
$$b=\sum_{(w',z)\neq(n'-1,y)}b_{w',z}<\frac{(\bar n-1)+(m+1)\bar r+2}{\bar r}l. $$	
For a fixed $\bar r$, it follows  $\frac{b}{b_{\min}}\rightarrow 1$ as $\bar n\rightarrow\infty$, where $b_{\min}=\frac{(\bar{n}-1)l}{{\bar{r}}}$. Hence our repair scheme is asymptotically optimal.   
\end{proof}

{Consider the $[n,k']$ RS code $RS(A,k')$ over $F$, where $k'>k$. Actually, the repair polynomials $\{g_{t,s}(x):t\in T_{w,y},s\in[0,\bar r-1]\}$ given in Construction \ref{cons2} for $RS(A,k')$ can also be used to repair failed node $\alpha_{w,y,j}$ of $RS(A,k)$. This is because ${\rm deg}(g_{t,s})\leq n-k'-1<n-k-1$ and $\mathrm{rank}_{B}(\{g_{t,s}(\alpha_{w,y,j}):t\in T_{w,y},s\in [0,\bar r-1]\})=l$ hold. 

In fact, the restriction that $\bar{r}$ is a composite number in Construction \ref{cons2} is not necessary.
Suppose $RS(A,k)$ over $F$ is defined as in Construction \ref{cons2}. If $\bar{r}\geq5$ is a prime number, then $\bar{r}'=\bar r-1$  is a  composite number. Let $\bar{k}'=\bar{n}-\bar{r}'$ and $k'=\bar{k}'u+v=(\bar{n}-\bar{r}')u+v>(\bar{n}-\bar{r})u+v=k$. Then the repair scheme for $RS(A,k')$ given in Construction \ref{cons2} is also suitable for $RS(A,k)$. Moreover, using the similar method as in Theorem \ref{thm6}, the RB can be estimated as follows. 

{\begin{corollary}\label{cor7}
If $\bar{r}\geq5$ is a prime number, i.e., $m=1$, then the RB of the repair scheme for $RS({A},k)$ given in Construction \ref{cons2} satisfies $\frac{b'}{b_{\mathrm{min}}}\rightarrow \frac{\bar{r}}{\bar{r}-1}\leq\frac{5}{4}$ as $\bar{n}\rightarrow \infty$.
\end{corollary}}
\begin{proof}
Denote $\bar{r}'=\bar{r}-1$. Since $\bar{r}\geq 5$ is a prime number, we know $\bar{r}'$ is a composite number. According to our previous analysis and Theorem \ref{thm6}, the repair scheme for $RS(A,k')$ is also suitable for $RS({A},k)$ and  the RB of  $RS({A},k)$ satisfies $\frac{b}{b'}\rightarrow 1$ as $\bar{n}\rightarrow \infty$, where $b'=\frac{(\bar{n}-1)l}{\bar{r}'}$. Since $b_{\mathrm{min}}=\frac{(\bar{n}-1)l}{\bar{r}}$, it follows that $\frac{b'}{b_{\mathrm{min}}}\rightarrow \frac{\bar{r}}{\bar{r}-1}$ as $\bar{n}\rightarrow \infty$.
\end{proof}

\section{CONCLUSION}\label{Sec4}
We present rack-aware RS codes with asymptotically optimal RB and smaller sub-packetization for all admissible parameters in this paper. In the future work, we expect to further reduce the sub-packetization level with optimal RB.

\appendix
\begin{proof}[Proof of Proposition \ref{pro11}]
Since $\alpha_{w,y,j}=\alpha^j\zeta^{d_{w,y}}$ and $o(\alpha)=u$, we have $g_{t,s}(\alpha_{w,y,j})=\zeta^{ut}\zeta^{usd_{w,y}}=(\zeta^{u})^{t+sd_{w,y}}.$ According to Proposition \ref{pro7}, ${\rm rank}_B\big(\{1,\zeta^u,\zeta^{2u},\dots,\zeta^{(l-1)u}\}\big)=l$ because $u\mid q-1$. Therefore, it suffices to prove $t+sd_{w,y}$ for $t\in T_{w,y},s\in[0,\bar r-1]$ can range over $[0,l-1]$. 

Let
\begin{align*}
T'=\{&a\in[0,l-1]:\\
&a=(0,\dots,0,a_{wm+y},\dots,a_{wm+y+m-1},0,\dots,0)\}.
\end{align*}
It is evident that $\bigcup_{t\in T_{w,y}} t+T'=[0,l-1]$. We next prove it actually has $T'=\{sd_{w,y}:s\in[0,\bar r-1]\}$.
For $e\in\mathbb{Z}$, define $$\bar{e}= \begin{cases}
e & \text{if}~ e\leq m \\
e \mod m & \text{otherwise}
			\end{cases}.$$
For any $a\in T'$, we have $$a=d_{w,y}\sum_{y'=0}^{m-1}a_{wm+y+y'}\big(\frac{1}{p_{y-1}}\prod_{e=y-1}^{y+y'-1}{p_{\bar{e}}}\big),$$ where $a_{wm+y+y'}\in [0,p_{\overline{y+y'}}-1]$. According to Theorem \ref{thm9}, $\sum_{y'=0}^{m-1}a_{wm+y+y'}(\frac{1}{p_{y-1}}\prod_{e=y-1}^{y+y'-1}{p_{\bar{e}}})$ can represent any integer in $[0,\bar r-1]$. Therefore, $T'=\{sd_{w,y}:s\in[0,\bar r-1]\}$. The proof is completed.   
\end{proof}

\begin{proof}[Proof of Proposition \ref{pro12}]
 Let $T'=\{sd_{n'-1,y}:s\in[0,\bar r-1]\}$. By the definition of $g_{t,s}(x)$, we need to prove $\mathrm{rank}_B\{(\zeta^{u})^{t+a}:t\in T_{n'-1,y},a\in T'\}=l$. For $ a\in T'$, following the discussion in Proposition \ref{pro11}, the multi-base expansion of $a$ is
$$a=(0,\dots,0,a_{(n'-1)m+y},\dots,a_{(n'-1)m+y+m-1}),$$ where $a_{(n'-1)m+y+y'}\in[0,p_{\overline{y+y'}}-1],y'\in[0,m-1]$. For any $t\in T_{n'-1,y}$, the muti-base expansion of $t$ is $$t=(0,\dots,0,t_{y},t_{y+1},\dots,t_{(n'-1)m+y-1},0,\dots,0).$$
 Therefore, $a+t,a\in T',t\in T_{n'-1,y}$ can represent any integer in $\prod_{j=0}^{y-1}p_j\cdot[0,l-1]$. Consequently, $\{(\zeta^{u})^{t+a}:t\in T_{n'-1,y},a\in T'\}=\{(\zeta^{u\prod_{j=0}^{y-1}p_j})^v:v\in[0,l-1]\}$. It remains to prove  the degree of the minimal polynomial of $\zeta^{u\prod_{j=0}^{y-1}p_j}:=\gamma$ over $B$ is $l$. 
 
Without loss of generality, we assume $p_1\leq p_2\leq\cdots\leq p_m$. Suppose the degree of the minimal polynomial of $\gamma$ is $t\ne l$, then $t\mid l$ and $\gamma^{q^t-1}=1$, i.e., $(q^l-1)\mid u\prod_{j=0}^{y-1}p_j(q^t-1)$. Therefore, $u\prod_{j=0}^{y-1}p_j(q^t-1)\geq q^l-1$, i.e., $u\prod_{j=0}^{y-1}p_j\geq \frac{q^l-1}{q^t-1}.$ However, since $y\in[m]$ and $p_1$ is the smallest divisor of $l$, we know $u\prod_{j=0}^{y-1}p_j\leq u\frac{\bar r}{p_1}$ and $t\leq \frac{l}{p_1}$. Thus, \begin{equation*}\label{eq10}
 	\big(u\prod_{j=0}^{y-1}p_j\big)-\frac{q^l-1}{q^t-1}\leq u\frac{\bar r}{p_1}-\frac{q^l-1}{q^{\frac{l}{p_1}}-1}\leq (q-1)\frac{\bar r}{p_1}-\sum_{i=0}^{p_1-1}q^{i\frac{l}{p_1}}.
 \end{equation*} To rule out the trivial cases, we may assume $\bar r\geq 2$, i.e., $p_1\geq 2$. Then $\sum_{i=0}^{p_1-1}q^{i\frac{l}{p_1}}\geq 1+q^{{\frac{l}{p_1}}}\geq 1+q^{{\frac{\bar r}{p_1}}}$ because $l=\bar{r}^{n'}$. Moreover, since $(q-1)x< 1+q^x$ for any integer $x$, it follows that $(q-1)\frac{\bar r}{p_1}-\sum_{i=0}^{p_1-1}q^{i\frac{l}{p_1}}\leq (q-1)\frac{\bar r}{p_1}-(1+q^{{\frac{\bar r}{p_1}}})<0$ which contradicts $u\prod_{j=0}^{y-1}p_j\geq \frac{q^l-1}{q^t-1}.$ Then the proof is finished.   
\end{proof}


\begin{thebibliography}{00}
\bibliographystyle{ieeetran}
\bibitem{facebook} S. Muralidhar \textit{et al.}, ``f4: Facebook's warm {BLOB} storage system," in \textit {Proc. 11th USENIX Symp. Operating Syst. Des. Implement. (OSDI),} Oct. 2014, pp. 383–398. 
\bibitem{google}Colossus: Successor to the google file system. \url{http://www.systutorials.com/3202/colossus-successor-to-google-file-system-gfs/.}
\bibitem{baidu} C. Lai \textit{et al}., ``Atlas: Baidu's key-value storage system for cloud data," in \textit {proc.  31st Symp. Mass Storage Syst.  Technol. (MSST),} 2015, pp. 1-14.
\bibitem{mic}C. Huang \textit{et al}., `` Erasure coding in windows azure storage," in \textit{ proc. 2012 USENIX Annu. Tech. Conf. (ATC), }Jun.  2012, pp. 15-26.
\bibitem{Gu} V. Guruswami and M. Wootters,``Repairing Reed-Solomon codes,”
\textit{IEEE Trans. Inf. Theory,} vol. 63, no. 9, pp. 5684–5698, Sep. 2017.
\bibitem{Dau1} H. Dau and O. Milenkovic, ``Optimal repair schemes for some families
of full-length Reed-Solomon codes,” in  \textit{Proc. IEEE Int. Symp. Inf. Theory (ISIT),} Jun. 2017, pp. 346–350.
\bibitem{Ye1}M. Ye and A. Barg, ``Explicit constructions of MDS array codes and RS codes with optimal repair bandwidth," in \textit{Proc.  IEEE Int. Symp.  Inf. Theory (ISIT),} Jul. 2016, pp. 1202-1206.
\bibitem{Ye2}  I. Tamo, M. Ye, and A. Barg, “Optimal repair of Reed-Solomon codes:
Achieving the cut-set bound,” in \textit{Proc. 58th Annual IEEE Symp. Founda
tions Computer Sci. (FOCS),} Oct. 2017, pp.  216-227.
\bibitem{Li}W.  Li, Z. Wang and H. Jafarkhani, ``On the sub-packetization size and the repair bandwidth of Reed-Solomon codes,"  \textit{IEEE Trans.  Inf. Theory,}, vol. 65, no. 9, pp. 5484-5502, Sep. 2019.
\bibitem{Vardy}A. Chowdhury and A. Vardy, ``Improved schemes for asymptotically optimal repair of MDS codes," \textit{IEEE Trans. Inf. Theory,}  vol. 67, no. 8, pp. 5051-5068, Aug. 2021, 
\bibitem{Hu1} Y. Hu, P. P. Lee, and X. Zhang, ``Double regenerating codes for hierarchical data centers,” in \textit{Proc.  IEEE Int. Symp.  Inf. Theory (ISIT),} Jul. 2016, pp. 245–249.
\bibitem{Hou1}  H. Hou, P. P. Lee, K. W. Shum, and Y. Hu, ``Rack-aware regenerating
codes for data centers,” \textit{IEEE Trans.  Inf. Theory,} vol. 65, no. 8, pp. 4730–4745, Aug. 2019.
\bibitem{Chen1}Z. Chen and A. Barg, ``Explicit Constructions of MSR Codes for Clustered Distributed Storage: The Rack-Aware Storage Model," \textit{IEEE Trans.  Inf. Theory,}  vol. 66, no. 2, pp. 886-899, Feb. 2020.
\bibitem{Jin}  L. Jin, G. Luo, and C. Xing, ``Optimal repairing schemes for Reed-
Solomon codes with alphabet sizes linear in lengths under the rack-aware
model,” \textit{ arXiv:1911.08016 [cs.IT],} Nov. 2019.
\bibitem{wang2023}J. Wang and  Z. Chen,  ``Low-access repair of Reed-Solomon codes in rack-aware storage" in \textit{Proc.  IEEE Int. Symp.  Inf. Theory (ISIT),} Jun. 2023, pp. 1142-1147.



\end{thebibliography}
\end{document}